\documentclass[conference]{IEEEtran}
\IEEEoverridecommandlockouts

\usepackage{amsthm}
\usepackage{mathtools}
\usepackage{algorithm}
\usepackage{breqn}
\usepackage{cite}
\usepackage{amsmath,amssymb,amsfonts}
\usepackage{algorithmic}
\usepackage{graphicx}
\usepackage{textcomp}
\usepackage{xcolor}
\def\BibTeX{{\rm B\kern-.05em{\sc i\kern-.025em b}\kern-.08em
    T\kern-.1667em\lower.7ex\hbox{E}\kern-.125emX}}

\begin{document}

\theoremstyle{definition}
\newtheorem{definition}{Definition}
\newtheorem{theorem}{Theorem}

\title{Differential Privacy for Protecting Private Patterns in Data Streams
\thanks{This work was funded by the Parrot Project (Research Council of Norway, project number 311197).}}

\author{\IEEEauthorblockN{He Gu,
Thomas Plagemann, Maik Benndorf and
Vera Goebel}
\IEEEauthorblockA{\textit{Department of Informatics} \\
\textit{University of Olso}\\
Oslo, Norway \\
{\{heg, plageman, maikb, goebel\}@ifi.uio.no}}
\and
\IEEEauthorblockN{Boris Koldehofe}
\IEEEauthorblockA{\textit{Bernoulli Institute} \\
\textit{University of Groningen} \\
Groningen, The Netherlands \\
b.koldehofe@rug.nl}
}

\maketitle

\begin{abstract}
Complex event processing (CEP) is a powerful and increasingly more important tool to analyse data streams for Internet of Things (IoT) applications. These data streams often contain private information that requires proper protection. However, privacy protection in CEP systems is still in its infancy, and most existing privacy-preserving mechanisms (PPMs) are adopted from those designed for data streams. Such approaches undermine the quality of the entire data stream and limit the performance of IoT applications. In this paper, we attempt to break the limitation and establish a new foundation for PPMs of CEP by proposing a novel pattern-level differential privacy (DP) guarantee. We introduce two PPMs that guarantee pattern-level DP. They operate only on data that correlate with private patterns rather than on the entire data stream, leading to higher data quality. One of the PPMs provides adaptive privacy protection and brings more granularity and generalization. We evaluate the performance of the proposed PPMs with two experiments on a real-world dataset and on a synthetic dataset. The results of the experiments indicate that our proposed privacy guarantee and its PPMs can deliver better data quality under equally strong privacy guarantees, compared to multiple well-known PPMs designed for data streams.
\end{abstract}


\section{Introduction}
As telecommunication techniques and sensor industries evolve, data-driven Internet of Things (IoT) applications become increasingly important, with a significant impact on economy and human life \cite{b6}. Valuable sensor data can in many cases reveal information that individuals regard as private and must be protected from misuse. Several studies (e.g., \cite{b2, b3}) have investigated privacy-preserving mechanisms (PPMs) for data streams and have delivered provable privacy guarantees and satisfactory data quality. Most of them operate on raw data tuples and basic events in data streams. However, for IoT applications built on complex event processing (CEP) systems, focusing privacy protection on private patterns rather than raw data tuples seems more profitable. It reduces the damage to data quality and hence provides better performance of IoT applications.

Consider the example of Taxi services in which cars equipped with GPS sensors resemble IoT devices. The GPS locations of the cars and passenger requests are used to match nearby taxi drivers with waiting passengers. Some passengers do not want their locations revealed when traveling to certain sensitive locations, such as their homes, bars, or hospitals. To hide trips related to sensitive locations, typical state-of-the-art PPMs would add sufficient noise to the GPS records of an entire trip. This action reduces the values of all GPS location tuples and the utility of all location-based services. Considering that a passenger only wants to hide proximity to sensitive locations, it is more efficient to protect only the data that reveal private patterns and to avoid adding noise to all other data. It helps maintain a higher quality of the entire data stream and enhances the usefulness for all services that depend on this data stream. As such actions only protect specific private patterns, we name them \textbf{pattern-level PPMs}. 

So far, only a few studies have been conducted to develop pattern-level PPMs and their theoretical foundation. We present in this paper a novel pattern-level differential privacy (DP) guarantee. Based on the presented DP, we propose two novel pattern-level PPMs that protect any private pattern against any query that requires a binary answer, i.e., a complex event is detected or not. We introduce a customized CEP engine that receives a specification of private patterns and raw data streams that may contain private information from data subjects. When being queried by IoT applications, the CEP engine protects the private patterns before answering. In detail, this paper makes the following contributions:
\begin{itemize}
    \item We propose a novel DP guarantee, which is named pattern-level $\epsilon$-DP (pattern-level DP).
    \item We propose two pattern-level PPMs that satisfy the pattern-level DP guarantee. We extend the upper performance boundary of privacy-aware IoT applications from non-pattern-level PPMs. Our evaluation confirms a superior performance of the proposed PPMs on two datasets compared to multiple non-pattern-level PPMs.
    \item We propose the first pattern-level PPM that offers sufficient granularity to both data subjects and IoT application developers, which provides adaptive privacy protection and output data quality.
\end{itemize}

We first introduce the related work in Section II. Then, we illustrate the system model, formulate the basic definitions, and state the problem for this work in Section III. In Section IV, we propose a novel DP guarantee for pattern-level PPMs, i.e., pattern-level DP. Based on this guarantee, two pattern-level PPMs are introduced in Section V. We implement our proposed PPMs and evaluate their performance in Section VI and conclude this paper in Section VII.

\section{Related Work}
Multiple novel theories and PPMs have been proposed for continuous data observation and data streams. Some works \cite{b2, b3} focus on adapting traditional PPMs on static databases to infinite dynamic data streams, while others \cite{b4} utilize the unique properties of data streams and their applications, e.g., reorder the detected events to protect private patterns. 

DP is one of the most reliable privacy guarantees. Although many studies successfully adapt DP to data streams, they mainly aim to transform the infinite stream into a static database and to further utilize DP on that pseudo-database\cite{b2}. 

Another well-known set of approaches rebuild the grounds of DP based on data streams with respect to the protection levels\cite{b16}, including event-level privacy\cite{b7}, user-level privacy\cite{b7}, and w-event privacy \cite{b8,b9}. Event-level privacy guarantees DP regarding each single event, user-level privacy guarantees DP regarding each single data provider, and w-event privacy guarantees DP regarding any event within a predefined window. Although these approaches deliver satisfactory results, they rarely emphasize the different characteristics of distinct data streams, which can be utilized to provide dedicated solutions and even superior performance. In the context of the Taxi example, these approaches provide similar privacy protections to a passenger regardless of the proximity to sensitive locations. Landmark privacy\cite{b10} makes a move in the direction of utilizing the different characteristics of data streams. It claims that, in reality, not all timestamps and data should be treated equally because some of them may contain significantly more private information or valuable information. Although it may seem similar to our approach, it does not take into account the connections between different data tuples, which makes it distinct from our work. Spatiotemporal events privacy \cite{b15} discusses the indistinguishability and proposes a novel DP guarantee regarding travel routes, which can be considered as a type of patterns. However, its definition of neighbors of its DP guarantee is distinct from our work, and its discussion about privacy is only limited to location-based data. Thus, it provides a significantly different study than ours.

\section{System  model and problem statement}

\subsection{System Model}
Our proposed privacy guarantee and PPM are based on a system model consisting of data subjects, a trusted CEP engine, and data consumers. According to the definitions of the General Data Protection Regulation (GDPR) \cite{b12}, we detail the requirements of these components as follows:

\begin{itemize}
    \item Data subjects supply data to the CEP engine and expect their private information to be protected according to their privacy requirements.
    \item The trusted CEP Engine is a middleware between data subjects and data consumers, providing privacy protections to data subjects while delivering the required data to data consumers.
    \item Data consumers, including both data controllers and data processors, consume data from the CEP engine. They require certain data quality and provide certain IoT services. They follow the \textit{honest-but-curious} threat model.
\end{itemize}

Additional assumptions for this work are presented after the following introduction of several basic definitions.

Given an infinite data stream $S^D$, we model it as an infinite tuple $S^D = (d_1, d_2, ...)$, where $d_i$ denotes the data provided by $S^D$ at timestamp $i$. Within a data stream $S^D$, any data tuple of our interest is considered an event. We can extract all events from a given data stream $S^D = (d_1, d_2, ...)$ in temporal order and form a new stream $S^E = (e_1, e_2, ...)$, where $e_i$ is the i-th extracted event, and $e_{i+1}$ is extracted after $e_i$. We name these new streams as \textbf{event streams}. When multiple data streams are given, we merge their corresponding event streams into one single event stream. Events from different event streams with the same timestamps can be ordered arbitrarily because their temporal order has no influence on any discussion in this paper.

Given an event stream $S^E = (e_1, e_2, ...)$, the combination of multiple events in temporal order can reveal certain useful information. We denote such a combination as a pattern $$P = seq(e_1, e_2, ..., e_m),$$ where the simplest pattern $P$ is an event. We define the events in the sequence as the elements of $P$, i.e., $e_i \in P$. Multiple patterns can form a higher-level pattern. For such a pattern, we collect all related events from all related lower-level patterns and merge them together so that any pattern can always be written in the form of a sequence of events.

Similarly, an event stream $S^E = (e_1, e_2, ...)$ can be abstracted into a \textbf{pattern stream} $S^P = (P_1, P_2, ...)$, where $P_i$ is the $i$-th detected pattern. $P_i$ and $P_j$ can denote the same type of patterns even if $i \neq j$. If $P_i \neq P_j$, they could also contain the same events. It makes their occurrences related to each other, and we define these patterns as \textbf{overlapping patterns} of each other. In the context of the Taxi example, the data streams contain all GPS records from all taxis and passengers. We collect the GPS tuples of interest and merge them into the event stream. The combination of several events forms a pattern, e.g., a taxi is now parking near a hospital, and the pattern stream only consists of patterns. An illustration of the relations between two data streams, an event stream, and a pattern stream is shown in Fig. \ref{fig5}. 

\begin{figure}[htbp]
\centerline{\includegraphics[width=0.45\textwidth]{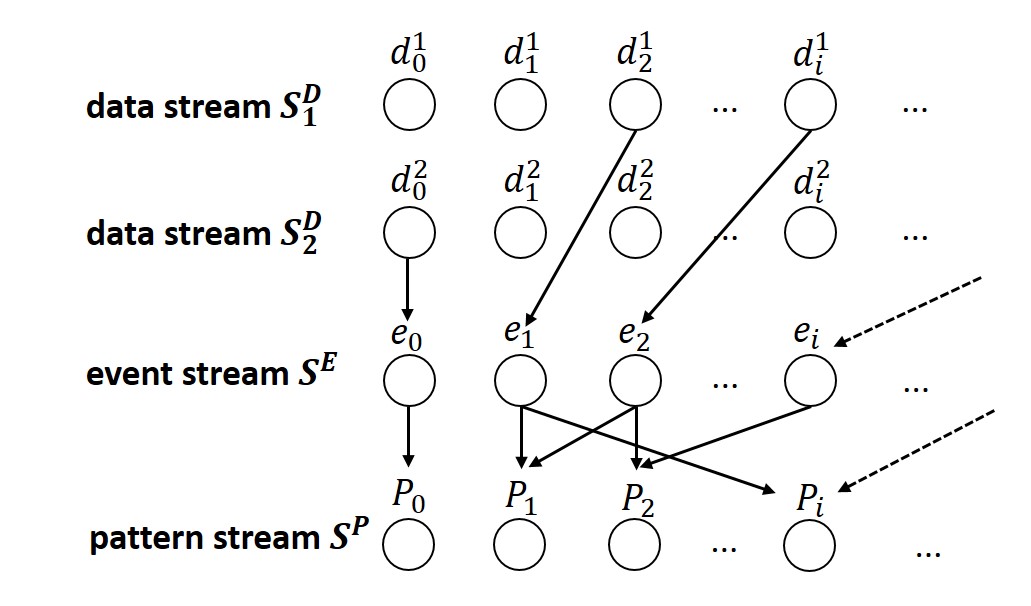}}
\caption{An illustration of the relations between data streams $S^D_1$ and $S^D_2$, event stream $S^E$, and pattern stream $S^P$.}
\label{fig5}
\end{figure}

Among the detected patterns, some may contain private information. These patterns are then defined as \textbf{private patterns}, while all others are \textbf{public patterns}. Data consumers are interested in the patterns that we call \textbf{target patterns}. A target pattern is either a private pattern or a public one. Considering the Taxi example, the pattern that a taxi is carrying a specific passenger to a hospital can be seen as a private pattern, while the pattern that this taxi is currently stuck in a traffic jam can be used as a target pattern which helps other drivers to optimize their routes. However, if the only traffic jam is happening around the hospital, then the target pattern and the private pattern overlaps and the private information needs to be protected when publishing the target pattern. Given the above definitions, our system is built under the following assumptions:
\begin{itemize}
    \item The CEP engine is trusted by data subjects so that it has access to raw data streams, including the private data.
    \item The CEP engine is trusted by data consumers and provides the answers to their queries.
    \item The queries to identify private patterns and target patterns are provided by data subjects and consumers to the CEP engine. The required data quality is defined by data consumers.
\end{itemize}

Fig. \ref{fig2} illustrates the execution procedure of our system model. During execution, we assume that data subjects provide raw data streams while data consumers consume complex events, namely patterns. The CEP engine takes raw data streams as input and delivers the results of queries sent by data consumers. 

\begin{figure}[htbp]
\centerline{\includegraphics[width=0.45\textwidth]{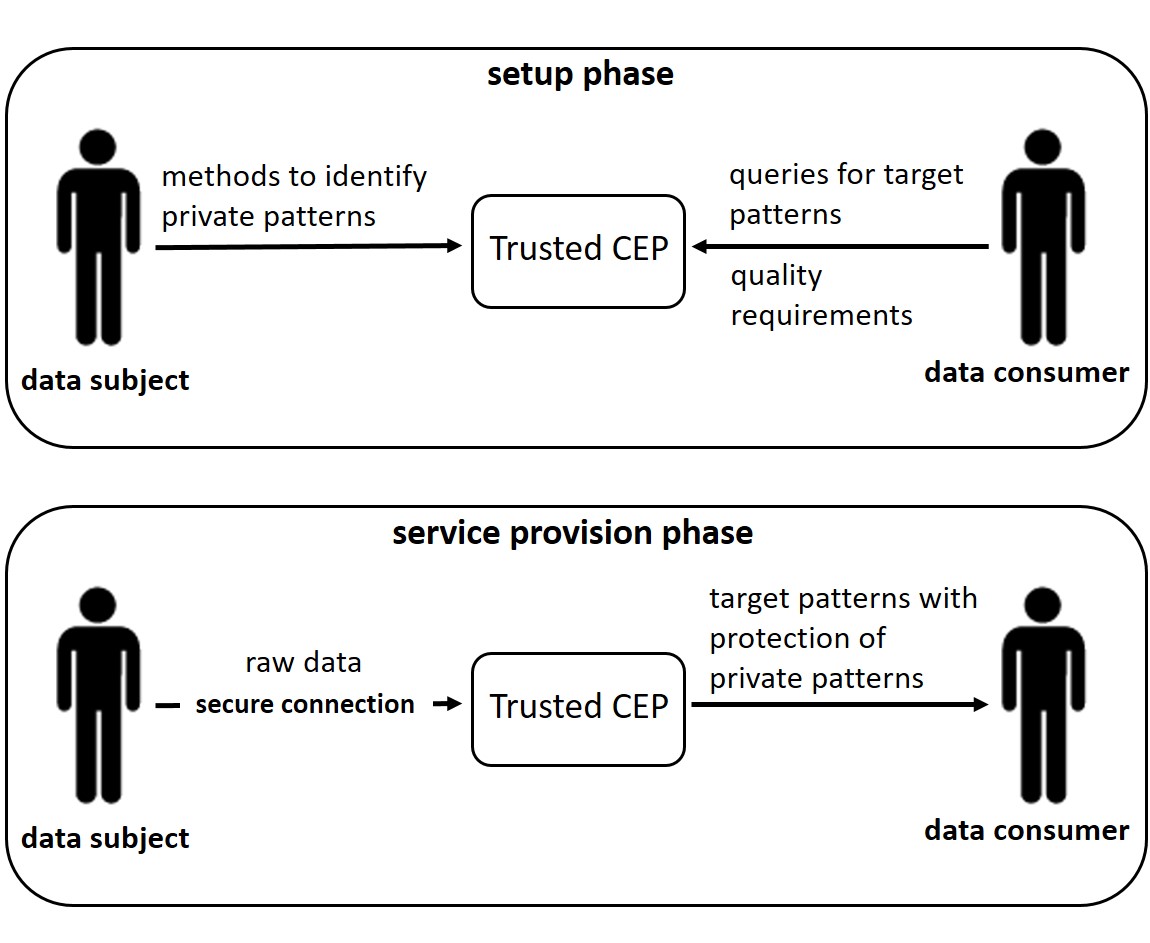}}
\caption{In the setup phase, data subjects define private patterns and data consumers define their data quality requirements and queries. In the service provision phase, the data subjects send raw data to the CEP engine, and the data consumers receive answers to their queries with privacy protection from the CEP engine. }
\label{fig2}
\end{figure}

\subsection{Problem Statement}
The objective of our proposed PPM is to protect predefined private patterns, while maintaining a predefined data quality. Privacy protection is achieved by establishing indistinguishability between private patterns and public patterns. Given a private pattern $P_{\mathit{pri}}$ and any target pattern $P_{\mathit{tar}}$, then for a continuous query that queries the existence of $P_{\mathit{tar}}$, we attempt to maximize the indistinguishability of the answers to this query regardless of the existence of $P_{\mathit{pri}}$. In the context of the Taxi example, our proposed PPMs aim to provide similar location-based services, e.g., similar traffic jam predictions, regardless of whether the passenger travels to sensitive locations or not.

In terms of data quality, an IoT application aims to detect as much target patterns as possible, which can be measured by the recall ($\mathit{Rec}$)
\begin{equation}
    \mathit{Rec} = \frac{\mathit{TP}}{\mathit{TP} + \mathit{FN}},
\end{equation}
and to reduce the number of false detections, which can be measured by the precision ($\mathit{Prec}$)
\begin{equation}
    \mathit{Prec} = \frac{\mathit{TP}}{\mathit{TP} + \mathit{FP}},
\end{equation}
where $\mathit{TP}$ is True Positive, $\mathit{FP}$ is False Positive, and $\mathit{FN}$ is False Negative.

However, evaluating only either precision or recall can be misleading, because identifying all patterns as target patterns can raise recall to $100\%$ while severely damaging precision and overall performance. Therefore, the metric of data quality can be a combination of both recall and precision:
\begin{equation}
    Q = \alpha \mathit{Prec} + (1-\alpha) \mathit{Rec},
\end{equation}
where $Q$ denotes the quality of the output data, and $\alpha$ is a hyperparameter predefined by data subjects and consumers.

We aim to minimize the decrease in data quality caused by a PPM. We use Mean Relative Error (MRE) to measure the decrease:
\begin{equation}
    \mathit{MRE}_Q = \frac{Q_{\mathit{ord}} - Q_{\mathit{PPM}}}{Q_{\mathit{ord}}},
\end{equation}
where $Q_{\mathit{ord}}$ denotes the ordinary data quality without applying any PPM, while $Q_{\mathit{PPM}}$ is the quality after employing a PPM.

By solving these two problems, our PPM can either (1) maximize data quality when given a fixed privacy budget, (2) or maximize privacy protection when given data quality requirements.

\section{Pattern-level $\epsilon$-differential privacy}
In this section, we propose a novel pattern-level privacy guarantee. It achieves $\epsilon$-DP with respect to a given pattern, and is hence named pattern-level $\epsilon$-DP (pattern-level DP).

Based on the definitions introduced in Section III, we first define a basic type of neighbors, i.e., \textbf{in-pattern neighbors}.

\begin{definition}
    Two patterns $P = seq(e_1, e_2, ..., e_m),$ and $P' = seq(e'_1, e'_2, ..., e'_m),$ of the same length are \textbf{in-pattern neighbors} if and only if (1) there exists a unique $i$ such that $P_i \neq P'_i$, and (2) for all $j \neq i$, $P_j = P'_j$ holds.
\end{definition}

In-pattern neighboring indicates that the only difference between two detected patterns is a basic event $e_i$. For clarification, we prepare ourselves by defining the pattern instances and the pattern types as follows.
\begin{definition}
    A pattern type $\mathcal{P}$ is a group of patterns specified by a given query $q$. All elements in $\mathcal{P}$ can be identified by $q$, and any pattern instance $P_i$ identified by $q$ is an element of $\mathcal{P}$, i.e., $P_i \in \mathcal{P}$.
\end{definition}

Based on this definition, we propose pattern-level neighboring.

\begin{definition}
    Given a predefined pattern type $\mathcal{P}$, and two infinite pattern streams $S^P = (P_1, P_2, ...)$ and $S^{P'} = (P'_1, P'_2, ...)$, then $S^P$ and $S^{P'}$ are \textbf{pattern-level neighbors} with respect to $\mathcal{P}$ if and only if for any integer $i$ such that $P_i \in \mathcal{P}$, (1) $P_i$ and $P'_i$ are in-pattern neighbors, and (2) for $j \neq i$, $P_j = P'_j$ holds.
\end{definition}

Similar to $\epsilon$-DP \cite{b11}, we now define the pattern-level $\epsilon$-DP, which guarantees $\epsilon$-DP with respect to a predefined type of patterns $\mathcal{P}$.

\begin{definition}
    Assume that $\mathcal{M}$ is a mechanism that takes a pattern stream $D$ as input and outputs a response $R$ that belongs to the group of all possible responses $\mathcal{R}$. Then we claim that $\mathcal{M}$ satisfies \textbf{pattern-level} $\boldsymbol{\epsilon}$\textbf{-DP} of a given type of patterns $\mathcal{P}$ (pattern-level DP of $\mathcal{P}$) if and only if for any pattern-level neighbors $S^P$ and $S^{P'}$ of $\mathcal{P}$ and any sets of response $\mathcal{R}_i \subseteq \mathcal{R}$, 
    $$\Pr[\mathcal{M}(S^P) \in \mathcal{R}_i] \leq 
    e^{\epsilon} \cdot \Pr[\mathcal{M}(S^{P'}) \in \mathcal{R}_i]$$ holds.
\end{definition}

Pattern-level DP of pattern type $\mathcal{P}$ only takes into account the events that constitute $P \in \mathcal{P}$ and only aims to provide privacy guarantees for $P$. Therefore, in practice, $P$ is the private pattern predefined by data subjects. Mechanisms that satisfy pattern-level DP intend to output similar results regardless of the existence of private patterns. Compared to other DP mechanisms introduced before in the related work, the pattern-level DP provides more customized privacy protection under different scenarios. It reduces the privacy budget $\epsilon$ assigned to events that are less important to protect private patterns and utilizes the budget more efficiently.

\section{Pattern-level PPMs}
In this section, we propose two pattern-level PPMs that satisfy pattern-level DP, i.e., a uniform PPM and an adaptive PPM. As they are named, we distribute the privacy budget $\epsilon$ uniformly or adaptively. 

These PPMs are built under the assumption that all answers to the queries are binary. We see the potential to further extend these PPMs so that they can process queries that require numerical or categorical answers, but this is not the subject of this paper. This assumption is made under the following considerations:

\begin{itemize}
    \item Several types of queries that are interested in patterns only require binary answers.
    \item Binary answers can be equivalent to categorical or numerical answers in some cases. In the Taxi example, drivers can be interested in the numbers of nearby passengers, which are numerical data. They may utilize them to predict popular areas. However, the exact number is not necessary, since their true intention is to know if this area is crowded, which can be answered in binary. 
\end{itemize}

\subsection{Uniform Pattern-level PPM}
Given the assumption we made, the randomized response is an intuitive approach to achieve our proposed pattern-level DP. We demonstrate it in Definition 5 based on event streams.

\begin{definition} [Randomized response in event streams]
Given a pattern $P = seq(e_1, e_2, ..., e_m)$, a mechanism $\mathcal{M}$ offers a randomized response if it takes the existence of events $I(e_i) \in \{0, 1\}$ as input and outputs responses $R_i \in \{0, 1\}$ for each pattern with possibility  
\begin{equation}
    \begin{cases*}
        \Pr(R_i = j | I(e_i) = j) = 1- p_i & \\
        \Pr(R_i = j | I(e_i) = k) = p_i, & \\
\end{cases*}
\end{equation}
where $j \neq k$.
\end{definition}

We study pattern-level DP characteristics of randomized response mechanisms, and it indicates that privacy budgets add up among different events within a pattern, as formulated in Theorem 1 and its proof.

\begin{theorem}
Given a pattern $P = seq(e_1, e_2, ..., e_m) \in \mathcal{P}$ and a randomized response mechanism $\mathcal{M}$ with $p_1, p_2, ..., p_n \leq \frac{1}{2}$ that takes $\boldsymbol{I}(\boldsymbol{e}) = (I(e_1), I(e_2), ...,I(e_n))$ as inputs and outputs responses $\boldsymbol{R} = (R_1, R_2, ..., R_n)$, $\mathcal{M}$ guarantees $\sum_{i:e_i \in P}{\ln{\frac{1-p_j}{p_j}}}$-pattern-level DP with respect to a given type of patterns $\mathcal{P}$.
\end{theorem}

\begin{proof}
For two pattern-level neighbors $S^P = (P_1, P_2, ...)$ and $S^{P'} = (P'_1, P'_2, ...)$ of a pattern type $\mathcal{P}$, if given any randomized response mechanism $\mathcal{M}$ as described above, then we see that  
\begin{dmath}
\frac{\Pr[\mathcal{M}(S^P) \in \mathcal{R}]}{\Pr[\mathcal{M}(S^{P'}) \in \mathcal{R}]} = \prod_{j:e_j \not\in P}{\frac{\Pr[\mathcal{M}(I(e_j)) = R_j]}{\Pr[\mathcal{M}(I(e_j')) = R_j]}} \cdot \prod_{j:e_j \in P}{\frac{\Pr[\mathcal{M}(I(e_j)) = R_j]}{\Pr[\mathcal{M}(I(e_j')) = R_j]}} \leq \prod_{j:e_j \in P}{\frac{1-p_j}{p_j}}.
\end{dmath}
It indicates that $\mathcal{M}$ guarantees $\ln{\prod_{j:e_j \in P}{\frac{1-p_j}{p_j}}}$-pattern-level DP, which can be rewritten as $\sum_{j:e_j \in P}{\ln{\frac{1-p_j}{p_j}}}$-pattern-level DP.
\end{proof}

\begin{figure}[htbp]
\centerline{\includegraphics[width=0.5\textwidth]{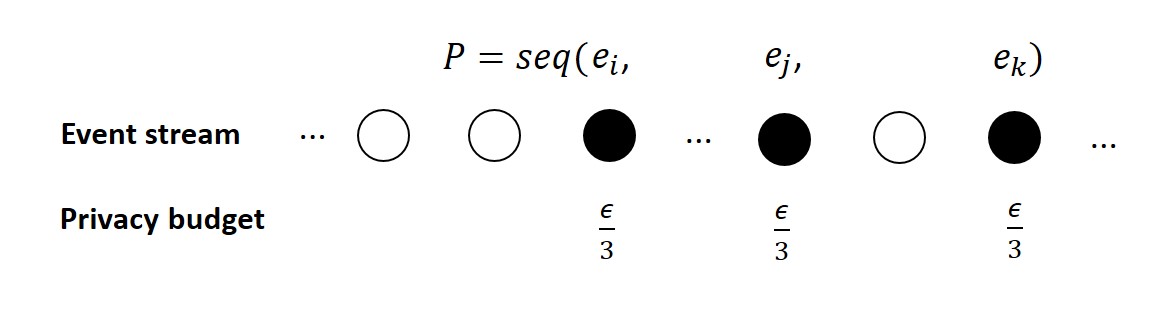}}
\caption{Uniform Privacy Budget Distribution}
\label{fig3}
\end{figure}

Pattern-level DP guarantees DP only with respect to a given pattern type $\mathcal{P}$. In practice, we consider $\mathcal{P}$ as the type of private patterns predefined by data subjects, since we aim to protect them. A basic approach is to distribute the given privacy budget $\epsilon$ evenly to each related pattern, as illustrated in Fig. \ref{fig3}.

Although overlapping patterns or repeating patterns will lead to more complicated situations, they do not damage or weaken our proposed pattern-level DP guarantee. The PPMs applied to these patterns can be independent, which leads to independent privacy budgets distribution for different patterns. It only brings more noise to the private information, which strengthens the protection of privacy.

\subsection{Adaptive Pattern-level PPM Based on Historical Data}
The uniform distribution may not provide optimal performance because some events are critical for detecting target patterns while containing little private information. An intuitive approach is to increase the privacy budgets for these events, leading to weaker privacy protection and higher data quality.

If given a private pattern $P =  seq(e_1, e_2, ..., e_m)$, we then denote the privacy budget distributed to the i-th event as $\epsilon_i = \ln{\frac{1-p_i}{p_i}}$. For a given total privacy budget $\epsilon$, $\sum_{i=1}^m{\epsilon_i} = \epsilon$ holds. 

The challenge of managing the budget distribution is equivalent to determining $p_i$. The $p_i$-s have a direct influence on our defined quality metrics $Q = \alpha Prec + (1-\alpha) Rec$, as they directly control the possibility of response to detect events and answer queries. Inspired by statistical learning, we introduce a bidirectional stepwise algorithm, i.e., Algorithm 1, to estimate the optimal $p_i$-s. The historical data utilized for the estimation are produced by the data subjects as defined in our system model. They trust our CEP engine so that they grant us the access to their historical data, even if these data may contain certain private information.

\subsection{Future Improvements}
We made an ideal assumption that, for a given pattern, all relevant events and patterns are ``perfectly" defined by data subjects and data consumers. However, it is a rigorous assumption since neither of these entities is expected to be privacy experts. Their classification of private patterns, public patterns, and relevant events can be inaccurate, raising the risk of privacy disclosure.

\begin{algorithm}
\caption{Bidirectional stepwise privacy budget distribution algorithm}
\begin{algorithmic}[1]
\STATE For any private pattern $P = seq(e_1, e_2, ..., e_m)$ and a given total privacy budget $\epsilon$, distribute evenly the budget to each element, namely, $\epsilon_i \gets \frac{\epsilon}{m}$;
\STATE Select the size of each step based on field experience. A suggestion is $\delta_\epsilon \gets \frac{m\epsilon}{100}$; 
\STATE Calculate data quality metric $Q = \alpha \mathit{Prec} + (1-\alpha) \mathit{Rec}$, and initialize $Q_1 = Q_2 = ... = Q_m = Q$;
 \WHILE{$\epsilon_1, \epsilon_2,...,\epsilon_m  \in [0, \epsilon]$ and $Q_{i:Q_i = \max{Q_i}} \geq Q$}
    \STATE $i \gets 1$; $Q \gets Q_{i:Q_i = \max{Q_i}}$;
    \WHILE{$i \leq m$}
        \STATE $\epsilon_i \gets \epsilon_i + \delta_\epsilon$; $\epsilon_{j:j \neq i} \gets \epsilon_{j} - \frac{\delta_\epsilon}{m}$; 
        \STATE $Q_i \gets \alpha \mathit{Prec} + (1-\alpha) \mathit{Rec}$; 
    \ENDWHILE
    \IF{$Q_{i:Q_i = \max{Q_i}} \geq Q$}
        \STATE $\epsilon_{i:Q_i= \max{Q_i}} \gets \epsilon_i + \delta_\epsilon$;  $\epsilon_{j:j \neq i} \gets \epsilon_{j} - \frac{\delta_\epsilon}{m}$.
    \ENDIF
 \ENDWHILE
\end{algorithmic}
\end{algorithm}

We do not emphasize this problem in this paper, but we see the potential to solve it. Similarly to our proposed adaptive PPM, we can estimate the correlations among events and patterns based on historical data, which enables us to reveal most of the latent relationships. It reduces the risks brought by the lack of privacy expertise.

\section{Evaluation}
\label{sec:eval}
We gather a real-world dataset from public resources to evaluate the proposed PPMs and compare their performance with that of other state-of-the-art techniques. In addition, because most existing real-world public datasets are not sufficient to demonstrate the advantages of the proposed PPMs, a synthetic dataset is also produced for evaluation.

\subsection{Experiment Setup}
\subsubsection{Dataset}
The utilized real-world dataset is named \textbf{Taxi}\footnote{https://www.microsoft.com/en-us/research/publication/t-drive-trajectory-data-sample/} \cite{b13, b14}. It contains the GPS records of 10357 taxis within Beijing. For each taxi, their unique id, GPS location, and timestamps were collected every 177 seconds, which corresponds to a distance of approximately 623 meters. The dataset is used in this paper because it is consistent with our example Taxi and has been used to evaluate other state-of-the-art DP theories \cite{b9}.

It is critical to specify the private patterns and target patterns during evaluation. We aim to monitor any taxi that enters the target pattern area while attempting not to reveal the actions of moving into the private pattern area. For the Taxi dataset, we randomly select 20\% GPS locations as the private pattern area and assign another 40\% as part of the target pattern area. We prefer to have some overlapping areas between target patterns and private patterns, because the evaluation is meaningful only if they are dependent and relevant to each other. Therefore, we randomly select 50\% of the private pattern area to become target pattern area, which leads to an overall 50\% target pattern area.

We generate the synthetic datasets and their private and target patterns with a random algorithm to test our proposed PPM in a complex scenario. We synthesized 1000 artificial datasets by repeating Algorithm 2 independently multiple times.

\begin{algorithm}[H]
\begin{algorithmic}[1]
\STATE Denote $20$ basic events as $e_1, e_2, ..., e_{20}$;
\STATE Randomly generate 20 numbers between 0 and 1 as the natural occurrence of $e_i$, i.e., $\Pr(e_i)$;
\STATE $m \gets 1$;
\WHILE{$m \leq 1000$}
    \STATE $n \gets 1$; set $L_m$ as an empty list;
    \WHILE{$n \leq 20$}
        \IF{a random number between 0 to 1 is smaller than $\Pr(e_n)$}
            \STATE Generate one $e_n$ and put it into list $L_m$;
        \ENDIF
    \ENDWHILE
\ENDWHILE
\STATE Collect the $1000$ $L_m$-s to be part of the synthetic dataset, where we regard each $L_m$ as a collection of events that detected within a window;
\STATE Among $20$ patterns $P_1, P_2, ..., P_{20}$, randomly select 3 as private ones and 5 as target ones;
\STATE Assign randomly 3 events to each of the $20$ patterns. If all three events are contained in one $L_m$, then their corresponding pattern is regarded as being detected.
\caption{Synthesis of artificial experimental data}
\end{algorithmic}
\end{algorithm}

\subsubsection{Performance}
Our proposed PPM is evaluated in terms of the precision and recall of detecting the target pattern. They are measured by MRE between data quality metrics without and with the application of PPM. Since the background knowledge gained from the datasets is insufficient, we set $\alpha = 0.5$ for the data quality metric $Q$, which emphasizes the precision and the recall equally.

From the related work, we select two typical algorithms of w-event DP, i.e., budget division (BD) and budget absorption (BA), as baseline \cite{b8}. Furthermore, since landmark privacy offers a similar privacy guarantee to our proposed PPM, we also compare its performance by evaluating its proposed adaptive algorithm\cite{b1}. The privacy budgets of BD, BA, and landmark privacy are converted from their original definitions to the one defined by pattern-level DP. The conversion is achieved by aggregating the original privacy budgets related to the predefined private pattern types. For different pattern types, an increase or a decrease of privacy budgets are both possible after a conversion. 

\subsection{Experiment Results and Analysis}
We implemented the experiments with Python. The result in Fig. \ref{fig4} shows that our pattern-level PPMs perform significantly better on synthetic datasets and relatively better on the real dataset Taxi. We did not specifically create a dedicated dataset which benifits our PPMs, but the advantage of them is expanded on the synthetic dataset because it simulates a more complex situation than the Taxi dataset.

\begin{figure}
\centerline{\includegraphics[width=0.45\textwidth]{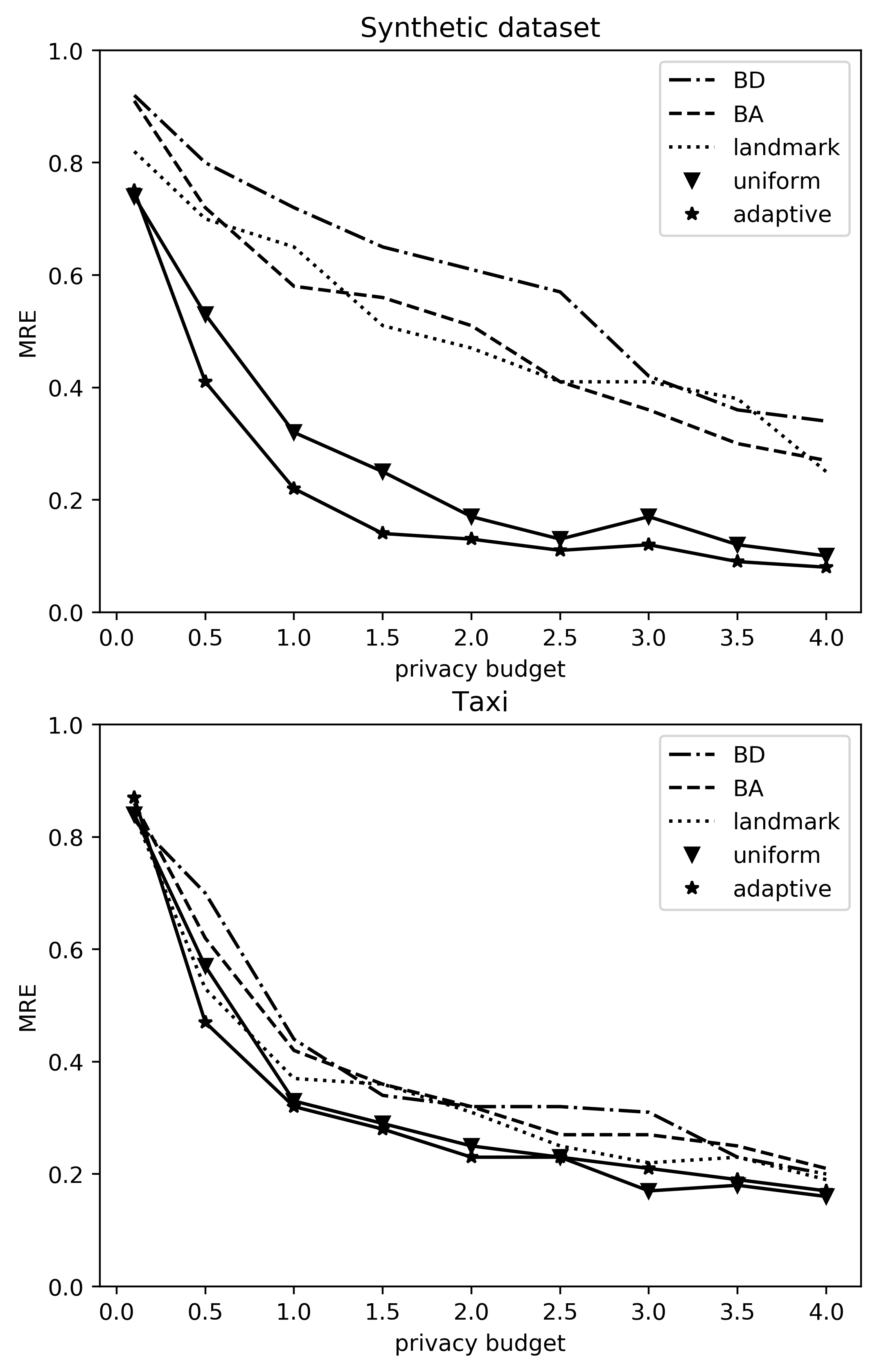}}
\caption{MRE with respect to different privacy budget $\epsilon$. The experiment datasets are Taxi and a synthetic dataset generated by Algorithm 2. }
\label{fig4}
\end{figure}

For the Taxi dataset, we see that the difference between the uniform and adaptive approaches is evidently smaller, as well as the remaining three algorithms. Recall that the test on Taxi is based on simple pattern types, i.e., GPS locations only. It indicates that detecting a pattern is almost identical to detecting a basic event. Therefore, the advantage of our proposed PPMs is significantly reduced, especially for lower privacy budgets.

\section{Conclusion}
With the growth of CEP-based IoT applications, the need for protecting the privacy of IoT data subjects increases. Non-pattern-level PPMs guarantee privacy with a superfluous cost of data quality. In order to reduce their redundancy, we propose in this paper the first pattern-level DP and multiple novel pattern-level PPMs. Our approach protects private patterns under equally strong privacy guarantees while providing significantly better data quality than non-pattern-level approaches. We showed that our approach surpasses other state-of-the-art ones by evaluating our PPMs on a synthetic and a real-world dataset.

Our approach extends the upper performance boundary for IoT-based PPMs and provides sufficient granularity for both privacy protection and data quality. For future work, we see the potential to strengthen the generality of our approach so that it becomes available for more CEP operators and IoT services.

\section*{Acknowledgment}

The authors want to thank Mikhail Fomichev, Matthias Hollick, and Rana Tallal Javed for their valuable feedback on an earlier version of this paper.

\vspace{12pt}

\end{document}